%% file: main.tex
\documentclass[letterpaper,UKenglish,coveref,autoref, 11pt]{article}
\usepackage[margin=1in]{geometry}

\usepackage{latexsym}
\usepackage{amsmath,amssymb,amsthm,mathtools}
\usepackage{amsfonts}
\usepackage{cite}
\usepackage[linesnumbered,lined,ruled,noend,vlined]{algorithm2e}
\usepackage{apxproof}
\usepackage{authblk}
\usepackage{lineno}
\date{\empty}

\input{macro}

\title{Constant Amortized Time Enumeration of Independent Sets for Graphs with Bounded Clique Number}
\author[1]{Kazuhiro Kurita}
\author[2]{Kunihiro Wasa}
\author[1]{Hiroki Arimura}
\author[2]{Takeaki Uno}

\affil[1]{IST, Hokkaido University, Sapporo, Japan\\
  \texttt{\{k-kurita, arim\}@ist.hokudai.ac.jp}}
\affil[2]{National Institute of Informatics, Tokyo, Japan\\
  \texttt{\{wasa, uno\}@nii.ac.jp}}

\begin{document}

\maketitle

\begin{abstract}
    In this study, we address the independent set enumeration problem. 
    Although several efficient enumeration algorithms and careful analyses have been proposed for maximal independent sets, 
    no fine-grained analysis has been given for the non-maximal variant. 
    From the main result, we propose an algorithm \EnumIS for the non-maximal variant that runs in $\order{q}$ amortized time and linear space, 
    where $q$ is the clique number, i.e., the maximum size of a clique in an input graph. 
    Note that \EnumIS works correctly even if the exact value of $q$ is unknown. 
    Despite its simplicity, \EnumIS is optimal for graphs with a bounded clique number, such as, triangle-free graphs, planar graphs, bounded degenerate graphs, locally bounded expansion graphs, and $F$-free graphs for any fixed graph $F$, where a $F$-free graph is a graph that has no copy of $F$ as a subgraph. 
\end{abstract}

\section{Introduction}
\label{sec:intro}

A subgraph enumeration problem is defined as follows: 
Given a graph $G$ and a constraint $\sig R$, 
the task is to output all subgraphs in $G$ that satisfy $\sig R$ without duplication. 
We call an algorithm for an enumeration problem an \name{enumeration algorithm}.
Enumeration problems have been widely studied, both in theory and practice, since 1950. 
The independent set enumeration problem is one of the central problems in the enumeration 
and several enumeration algorithms have been proposed for the maximal or maximum independent set enumeration problem~\cite{Tsukiyama:Ide:SICOMP:1977,Alessio:Roberto:SPIRE:2017,Beigel:SODA:1999,Eppstein:JGAA:2003}. 
In particular, theoretically efficient algorithms are developed by restricting the class of input graphs, 
e.g., chordal graphs~\cite{Okamoto:JDA:2008,Leung:JA:1984}, circular arc graphs~\cite{Leung:JA:1984}, bipartite graphs~\cite{Kashiwabara:JA:1992}, and claw-free graphs~\cite{Minty:JCT:1980}. 
Another important object to be enumerated is a clique, that is, the independent set of the compliment graph of a given graph. 
Several efficient algorithms exist for the maximal enumeration problem~\cite{Tomita:TCS:2006,Makino:Uno:SWAT:2004,Alessio:Roberto:ICALP:2016}. 
Since every non-maximal independent set is a subset of some maximal ones, 
using the above results, 
we can find all non-maximal solutions. 
However, it is difficult to avoid to output duplication efficiently. 
Moreover, even though there are many efficient algorithms for maximal variant, no fine-grained analysis has been given for non-maximal variant. 
Thus, this study aims to develop an efficient algorithm for the independent set enumeration problem, 
which also demands the output of non-maximal solutions.

Generally, the number of solutions to an enumeration problem is exponential in the size of the input. 
If the number of solutions is much smaller, then
it becomes unsuitable to evaluate the efficiency of an enumeration algorithm by the size $n$ of an input only since this can sometimes yield a trivial bound, such as $\order{2^n}$ time. 
In this paper, 
we evaluate the efficiency of an enumeration algorithm by \emph{both} the size $n$ of the input and the number $M$ of solutions. 
We call this analysis \name{output sensitive analysis}~\cite{Johnson:Yannakakis:Papadimitriou:IPL:1988}. 
Let $\mathcal{A}$ be an enumeration algorithm;  
$\mathcal{A}$ is an \name{output polynomial time algorithm} if the algorithm runs in $\order{poly(n, M)}$ time. 
$\mathcal{A}$ is a \name{polynomial amortized time algorithm} if the running time is bounded by $\order{M\cdot poly(n)}$, 
that is, $\mathcal{A}$ runs in $\order{poly(n)}$ time per solution on average.
$\mathcal{A}$ runs in $\order{poly(n)}$ \name{delay} if the interval between two consecutive solutions can be bounded by $\order{poly(n)}$ time and  the preprocessing time and postprocessing time are also bounded by polynomial in $n$.
Note that an polynomial amortized time algorithm does not guarantee that the maximum interval between two consecutive outputs, called the \name{delay}, is polynomial. 
So far,  several frameworks and a complexity analysis technique have been proposed for developing efficient enumeration algorithms~\cite{Avis:Fukuda:DAM:1996,Conte:Uno:STOC:2019,Cohen:Kimefeld:Sagiv:JCSS:2008,Conte:Grossi:Marino:Versari:SIAM:JoDM:2019,Uno:WADS:2015}. 
These frameworks have been used to develop several efficient output-sensitive enumeration algorithms, especially for sparse input graphs~\cite{Wasa:Arimura:Uno:ISAAC:2014,Wasa:COCOON:2018,Kazuhiro:ISAAC:2018,Alessio:Roberto:ICALP:2016,Manoussakis:TCS:2018,Bonamy:Defrain:Heinrich:Raymond:STACS:2019}. 
However, as will be shown later, 
simply applying the above results cannot generate an efficient algorithm for the independent set enumeration problem.

\begin{figure}
    \centering
    \includegraphics[width=0.6\textwidth]{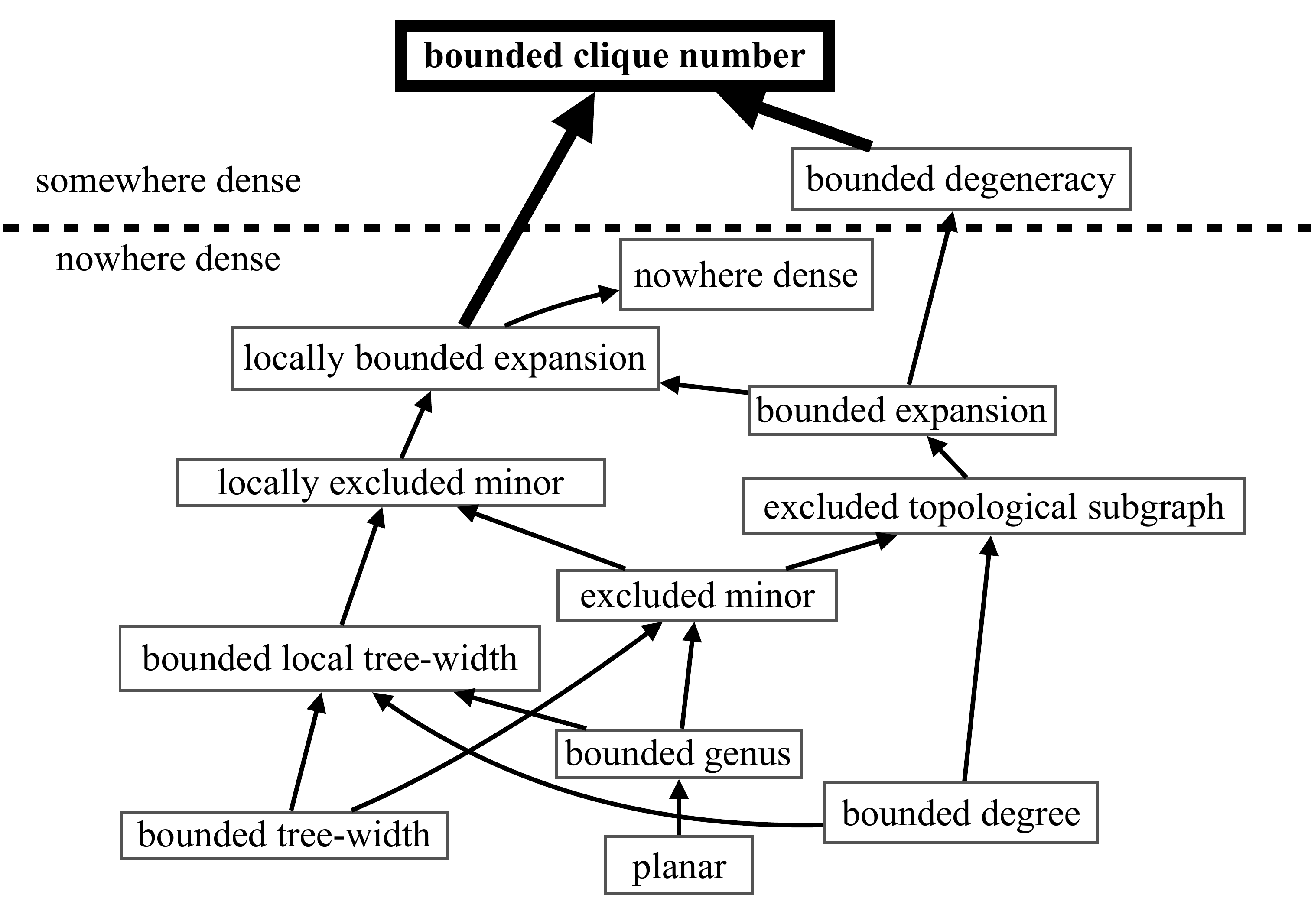}
    \caption{The inclusion relation between graph classes~\cite{Grohe:FSTTCS:2013}. 
    In this map, an arrow goes from a graph class to its super class, i.e., graphs with bounded clique number is a super class of bounded degeneracy graphs. }
    \label{fig:bip}
\end{figure}

\subsection{Main results} 
In this paper, we focus on $K_q$-free graphs, which are graphs that have no cliques with size $q$ as subgraphs.  
Note that every graph is $K_q$-free for some $q$ (e.g., $q = n+1$). 
In addition, 
it is known that if a graph $G$ does not have a clique with size $q$, 
then $G$ belongs to some sparse graph class, such as,
triangle-free graphs ($q = 3$),
planar graphs ($q = 5$ since they do not have both $K_5$ and $K_{3, 3}$),
locally bounded expansion graphs ($q = f(1, 0) + 1$ for some function $f(\cdot)$), 
bounded degenerate graphs ($q$ is at most the degeneracy plus one),
$F$-free graphs for some fixed graph $F$ ($q$ is the size of $F$), 
etc (Figure~\ref{fig:bip}), where an $F$-free graph is a graph that has no copy of $F$ as a subgraph. 
Using the main result of this paper, 
we propose an algorithm \EnumIS called for the independent set enumeration problem
that runs in $\order{q}$ amortized time with linear space. 
\EnumIS is optimal for these graph classes. 
Note that a complete graph with $q$ vertices contains any graph with $q$ vertices as a subgraph, 
thus if $q=\size{F}$, then \EnumIS is optimal for $F$-free graphs. 
We emphasize that \EnumIS works correctly even if the exact value of $q$ is unknown.

\EnumIS is simple \name{binary partition}. 
The algorithm starts with $(G, S=\emptyset)$, where $G$ is an $n$-vertex graph. 
First, \EnumIS outputs $S$ and computes a vertex sequence $(v_1, \dots, v_n)$ of $G$, sorted by a smallest-last ordering~\cite{Matula:Beck:JACM:1983}. 
Next, \EnumIS generates $n$ pairs, made of a subsolution and its corresponding graph  
$S_1 = S \cup \set{v_1}$ and $G_1 = G \setminus N[v_1]$, 
$S_2 = S \cup \set{v_2}$ and $G_2 = G \setminus (\set{v_1} \cup N[v_2])$, $\dots$, and 
$S_n = S \cup \set{v_n}$ and $G_n = G \setminus (\set{v_1, \dots,  v_{n-1}} \cup N[v_n])$. 
Then, for each pair,  \EnumIS makes recursive calls and repeats the above operations. 
We call this generation step an \name{iteration}. 
It can be easily shown that this algorithm runs in $\order{\Delta^2}$ amortized time 
since each iteration has new $n$ child iterations and needs $\order{n\Delta^2}$ time to generate all the children.
However, this naive analysis is too pessimistic since  
the number of vertices whose degree are $\Delta$ may be small. 
For example, if $G$ is a star with $n-1$ leaves that has a vertex with degree $n-1$, 
then the number of solutions is $\order{2^{n-1}}$. 
The first iteration of the algorithm requires $\order{n^2}$ time to generate all the children, 
while other iterations require $\order{\size{ch}}$ time since any subset of the remaining vertices makes a new solution, where $ch$ is the set of child iterations. 
Thus, the total time complexity is $\order{M-1 + n^2}$, i.e., \EnumIS runs in  $\order{1}$ time per solution on average. 
As described above, if $G$ has many vertices with small degree, 
then the simple analysis appears not to be tight. 
Conversely, if $G$ has no vertices with small degree, then $G$ has a large clique from Tur{\'a}n's theorem~\cite{Turan:1941}. 
From this observation, we focus on the clique number of an input graph to give a tight complexity bound. 

The analysis is based on the \name{push out amortization} technique~\cite{Uno:WADS:2015}, however,
it is difficult to directly apply the technique to the proposed algorithm. 
To apply this technique, we focus on the size of the input graph for an iteration. 
If an input graph is sparse, then there are many iterations with small input graphs, e.g., the size is constant,  
hence the sum of the computation time for these iterations dominates the total computation time for \EnumIS. 
In particular, we regard a graph as a small graph if the graph has at most $2q$ vertices. Otherwise, the graph is large. 
Surprisingly, the algorithm correctly works even if the location of the boundary is unknown, that is, the size of a maximum clique. 
In addition, by using \name{run-length encoding}, we show that \EnumIS uses only linear space in total. 
Due to space limitations, all proofs are given in the Appendix.

\section{Preliminaries}
\label{sec:prelim}
Let $G = (V(G), E(G))$ be a simple undirected graph, i.e.,  
$G$ has no self loops and multiple edges, 
where $V(G)$ and $E(G)$ are the set of vertices and edges of $G$, respectively. 
Let $n$ denote the number of vertices in $G$ and $m$ the number of edges in $G$. 
Let $u$ and $v$ be vertices in $G$;  
$u$ and $v$ are \name{adjacent} if $e = \set{u, v} \in E(G)$. 
We denote by $N_G(v)$ the set of the adjacent vertices of $v$ in $G$. 
We call $u$ a \name{neighbor} of $v$ in $G$ if $u \in N_G(v)$ 
and an edge $e = \set{u, v}$ an \name{incident edge} of $v$. 
$\size{N_G(v)}$ is the \name{degree} $d_G(v)$ of $v$. 
The degree of $G$ is the maximum degree of $v \in V(G)$. 
If there is no confusion, we can drop $G$ from the notations. 
A set of vertices $S$ of $G$ is an \name{independent set} if $G[S]$ has no edges, that is, 
for any pair of $u, v \in S$ is not included in $\set{u, v} \in E(G[S])$. 
Let $U$ be a vertex subset of $V$ and $G[U] = (U, E[U])$ be the subgraph of $G$ induced by $U$, 
where $E[U] = \inset{e \in E}{u, v \in U}$. 
For simplicity, 
we write $G \setminus \set{v} := G[V \setminus \set{v}]$, 
and $G \setminus e := (V, E \setminus \set{e})$. 
A graph $G$ is a \name{complete graph} if for any distinct pair $u, v \in V$, $\set{u, v} \in E$ and  
a set of vertices $S$ of $G$ is a \name{clique} if $G[S]$ is a complete graph. 
$K_n$ is expressed as a complete graph with $n$ vertices.  
$G$ is said to be \name{$K_q$-free} if $G$ has no clique with $q$ vertices. 
In this paper, we consider the following enumeration problem: 
Given an undirected graph $G$, then output all independent sets in $G$ without duplication.

\section{The proposed algorithm}
\label{sec:algoq2}

\begin{algorithm}[t]
    \caption{An $\order{q}$ amortized time enumeration algorithm \texttt{EIS} for independent sets, where $q$ is the clique number of $G$. }
    \label{algo:eis}
   \Procedure(\tcp*[f]{$G = (V, E)$: An input graph}){\EnumIS{$G$}}{
       \RecEIS{$G, \emptyset$}\;
    }
    \Procedure{\RecEIS{$G, S$}}{
        Output $S$\;
        \For(\tcp*[f]{$v$ has the minimum degree in $G$. }\label{algo:one:pick}){$v \in V$}{
            \RecEIS{$G[V \setminus N[v]], S \cup \set{v}$}\label{algo:1}\;
            $G \gets G\setminus \set{v}$\label{algo:0}\;
        }
    }
\end{algorithm}

In this section, 
we present a recursive enumeration algorithm \EnumIS based on \name{binary partition}, as shown in Algorithm~\ref{algo:eis}. 
Binary partition is a framework for developing enumeration algorithms. 
We provide a high level description of our proposed algorithm \EnumIS. 

Let $\sig S$ be the solution space of the independent set problem for a given graph $G$. 
For each recursive call $X$, called an \name{iteration}, 
$X$ is associated with a graph $G$ and a solution $S$. 
\EnumIS first outputs $S$ on the iteration. 
Next, it picks a vertex $v$ from $V$ and 
partitions the current solution space $\sig S$ into two distinct subspaces; 
one consists of solutions containing $v$  
and the other consists of solutions that do not contain $v$. 
Then, \EnumIS makes a new iteration $Y$ that receives $G\setminus N[v]$ and $S\cup\set{v}$. 
We call $Y$ a \name{child iteration} of $X$ and $X$ the \name{parent iteration} of $Y$. 
When backtracking from $Y$, \EnumIS removes $v$ from $G$, picks a new vertex, and makes a new child iteration $Y'$. 
Each iteration repeats the above procedure for all vertices in the input graph of the iteration. 
\EnumIS builds a recursion tree $\sig T = (\sig V, \sig E)$, 
where $\sig V$ is the set of iterations and $\sig E$ is given by the parent-child relation among $\sig V$. 
$X$ is referred to as a \name{leaf iteration} if $X$ has no child iterations. 
Otherwise, it is called an \name{internal iteration}. 
Let $G(X)$ and $S(X)$ be the input graph and the input independent set of $X$, respectively. 
From the construction of \EnumIS, we can obtain the following theorem. 

\begin{theorem}
\label{theo:enum}
    Let $G$ be a graph. 
    \EnumIS enumerates all independent sets in $G$ without duplication. 
\end{theorem}

\begin{proof}
    We first show \EnumIS outputs all solutions by induction on the size of a solution $S$. 
    We assume that all the solutions whose size are at most $k-1$ are outputted. 
    Let $k$ be the size of $S$ and 
    $S'$ be a vertex set $S \setminus \set{s}$, where $s \in S$.
    Note that any subset of $S$ is also an independent set, 
    and thus, $S'$ is an independent set. 
    From the assumption, there is an iteration $X$ which outputs $S'$. 
    If $G(X)$ contains $s$, then \EnumIS outputs $S$. 
    Otherwise, there is the lowest ancestor iteration $Z$ such that 
    removes $s$ from $G(Z)$ in Line~\ref{algo:0}. 
    Let $Y$ be an iteration such that $S(Y) = S(Z) \cup \set{s}$. 
    Since $S \cap G(X)$ is an independent set in $G$, 
    there is a descendant iteration of $Y$ which outputs $S$, i.e., all solutions are outputted. 
    
    Next, we show \EnumIS does not output duplicate solutions. 
    Let $X$ and $Y$ be two distinct iterations. 
    We assume that the both output the same solution. 
    Let $Z$ be the lowest common ancestor of $X$ and $Y$. 
    We assume that $Z \neq X$ and $Z \neq Y$. 
    Otherwise, the output of $X$ differs from the one of $Y$ from the construction of Algorithm~\ref{algo:eis}, 
    and this contradicts the assumption. 
    Let $z$ be a vertex picked in $Z$ such that $z \in S(X)$. 
    Again, from the construction of the algorithm, 
    $Y$ does not contain $z$. 
    Hence, this contradicts the assumption. 
\end{proof}

Note that Theorem~\ref{theo:enum} holds for any ordering of the vertices in an iteration. 
Hence, as in Line~\ref{algo:one:pick}, we employ the following simple picking ordering: 
pick a vertex with minimum degree. 
This ordering is known as \name{smallest-last ordering} or \name{degeneracy ordering}~\cite{Matula:Beck:JACM:1983}.  
Note that a smallest-last ordering is not unique since there are several vertices with the minimum degree. 
Thus, hereafter, we fix some deterministic procedure to uniquely determine the smallest-last ordering of $G(X)$.

\section{Time complexity}

\newcommand{\mybar}[1]{\overline{#1}}
\newcommand{\pushed}[1]{\rho(#1)}

In this section, we analyze the time complexity of \EnumIS. 
In the following, 
we restrict an input graph to be $K_q$-free. 
Note that every graph is $K_q$-free for some $q$ (for example, $q = n+1$), hence, the result is true for all graphs.
We first give a brief overview of our time complexity analysis, as shown in Figure~\ref{fig:tree}. 
The left part of the figure is the recursion tree $\sig T$ made by \EnumIS. 
In our analysis, we push a part of the computational cost of an iteration to its child iterations. 
The remaining cost is received by the iteration itself. 
The key point is to use different distribution rules from the gray area and the white area. 
The boundary between these areas is defined by the size of the input graph of an iteration. 
More precisely, 
the gray area contains iterations whose input graphs have less than or equal to $2q$ vertices (Sect.~\ref{subsec:time:small:case}), 
and the white area contains iterations whose input graphs have more than $2q$ vertices (Sect.~\ref{subsec:time:large:case}). 
This boundary on $\sig T$ gives a sophisticated time complexity analysis. 
A simple amortization analysis can be applied in the gray area. 
In the white area, 
the push out amortization technique~\cite{Uno:WADS:2015} is used to design the cost distribution rule. 

The right part shows how to push the computational cost from the iteration of the white area to that of the gray area. 
Gray rectangles represent costs pushed to the child iterations, 
while white rectangles represent costs received by the iteration itself. 
By using the push out amortization technique, we can show that an iteration $L_i$ receives only $\order{T(L_i)}$ computational time from its parent, 
where $T(L_i)$ is the computation time of $L_i$. 
That is, the delivered cost $\order{T(L_i)}$ does not worsen the time complexity of $L_i$. 

\begin{figure}[t]
    \centering
    \includegraphics[width=0.9\textwidth]{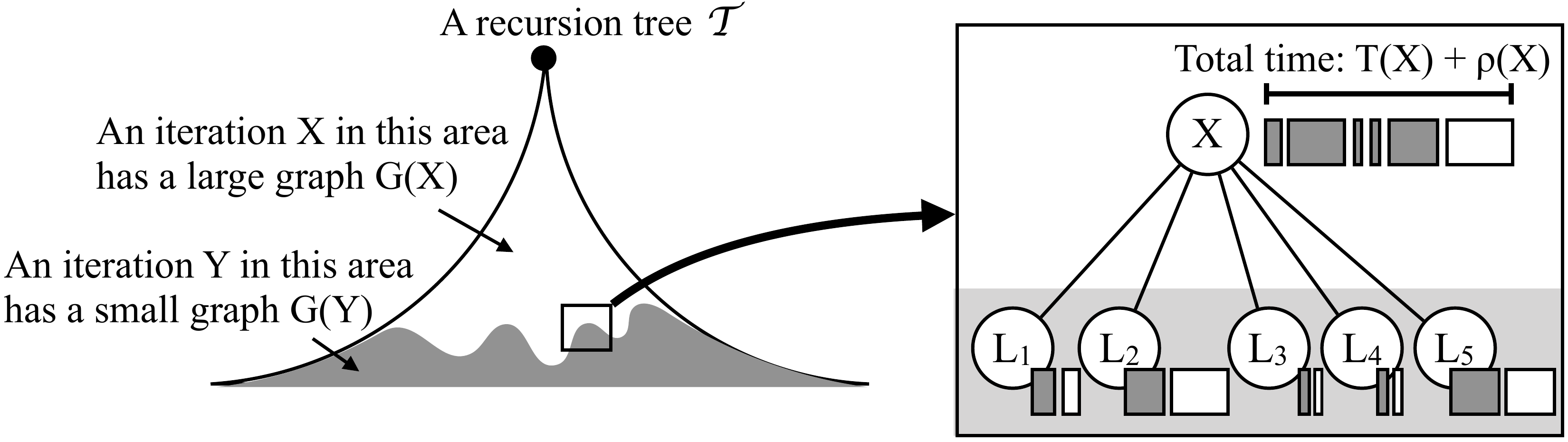}
    \caption{High level overview of our time complexity analysis. In the figure on the right, gray rectangles and a white rectangle represent the computational cost, and the sum of all rectangles means $T(X) + \rho(X)$. $X$ receives a white rectangle and gray rectangles are pushed out to each iteration.}
    \label{fig:tree}
\end{figure}

In the following, 
we provide the detailed description of our analysis. 
Assume that we use an adjacency matrix for storing the input graph.  
For simplicity, we write $V(X) = V(G(X))$ and  $E(X) = E(G(X))$. 
Let $n_X = \size{V(X)}$, $m_X = \size{E(X)}$, $ch(X)$ be the set of children of $X$, and 
$\sig T(X)$ be the recursive subtree of $\sig T$ rooted $X$. 
The next lemma is easy but plays a key role in this section: 

\begin{lemma}
\label{lem:smalltime}
    The total computation time of $\sig T(X)$ can be bounded by $\order{\size{\sig T(X)}n_X}$. 
\end{lemma}

\begin{proof}
    Let $\sig V(X)$ be the set of iterations on $\sig T(X)$. 
    For each iteration $Y$ in $\sig V(X)$, since each picked vertex $v$ on $Y$ generates a new child iteration,  
    $Y$ needs $\order{\size{ch(Y)}n_Y}$ time, 
    and thus, the total time of \EnumIS for $\sig T(X)$ is $\order{\sum_{Y \in \sig V(X)} \size{ch(Y)} n_Y}$. 
    Note that $\sum_{Y \in \sig V(X)} \size{ch(Y)} = \order{\size{\sig V(X)}}$. 
    In addition,  each iteration has the corresponding solution and $n_X \ge n_Y$ for any descendant iteration $Y$ of $X$. 
    Hence, the total time complexity is $\order{\size{\sig T(X)}n_X}$. 
\end{proof}

\subsection{Case: $n_X \le 2q$}
\label{subsec:time:small:case}

From Lemma~\ref{lem:smalltime}, 
if $X$ satisfies $n_X \le 2q$, 
then the total time complexity of $\sig T(X)$ is $\order{q}$ time on average. 
Note that for any descendant iteration $Y$ of $X$, 
since $n_Y \le n_X$, if $n_X \le 2q$, then $n_Y \le 2q$. 
Thus, the time complexity of $Y$ is also $\order{q}$ time on average if $X$ satisfies the condition.

\subsection{Case: $n_X > 2q$}
\label{subsec:time:large:case}

In this subsection, 
we use the \name{push out amortization}~\cite{Uno:WADS:2015} to analyze the case $n_X > 2q$. 
This is one of the general techniques for analyzing the time complexity of enumeration algorithms. 
Intuitively, 
if an enumeration algorithm satisfies the \name{PO condition}, 
the total time complexity of the algorithm can be bounded by the sum of the time complexity of leaf iterations with small time complexity. 
The PO condition is defined as follows: 
For any internal iteration $X$, 
$\mybar{T}(X) \ge \alpha T(X) - \beta(\size{ch(X)} + 1)T^*$. 
Here, 
$T^*$ represents the maximum time complexity among the leaf iterations, 
$\alpha > 1$ and $\beta \ge 0$ denotes some constants, 
$T(X)$ represents the time complexity of $X$, and 
$\mybar{T}(X)$ is the total computation time of child iterations of $X$. 
Note that if $T^*$ is large, then each internal iteration can more readily push its computation time out to its child iterations more.  

We first explain the outline of the push out amortization. 
In \cite{Uno:WADS:2015}, Uno gives a concrete computation time distribution rule for this amortization. 
Let $\pushed{X}$ be a computation time that is pushed out from the parent of $X$ and toward $X$. 
Hence, $X$ now has  $T(X) + \pushed{X}$ as its total computation time. 
To achieve $\order{T^*}$ time per solution on average, 
the computation time of $T(X) + \pushed{X}$ is delivered as follows: 
(D1) $X$ receives $\hat{T}(X) = \frac{\beta(\size{ch(X)} + 1)T^*}{(\alpha - 1)}$ and %
(D2) each child iteration $Y$ of $X$ receives 
the remaining computation time of $\pushed{Y} = \left(T(X) + \pushed{X} - \hat{T}(X)\right)\frac{T(Y)}{\mybar{T}(X)}$. 
In reality, 
since the sum of the number of child iterations of all iterations in $\sig T$ does not exceed the number of solutions, 
each iteration receives $\order{T^*}$ as (D1) on average. 
In addition, if the algorithm satisfies the PO condition, then $\pushed{Y} \le T(X) / (\alpha -1)$. 
The reader should refer to~\cite{Uno:WADS:2015} for more details. 
In the following,  we show that \EnumIS satisfies the PO condition. 
For the following discussion, we introduce some notations. 
Let $V_{i:j} = \inset{v_k \in V}{i \le k \le j}$ and  
$v_i$ be a vertex with minimum degree in $V_{i:n}$. 
We denote the subgraph of $G$ induced by $V(X) \setminus V(X)_{1:i - 1}$ by $G_i(X)$. 

We first consider the number of vertices of a child iteration. 
After picking a vertex with minimum degree in $G(X)$, $X$ removes its neighborhood from $G(X)$.  
Thus,  for each $1 \le i \le n_X$, 
the size of the input graph for the $i$-th children $Y_i$ of $X$ is $(n_X - i + 1) - d_{G_i(X)}(v_i)$. 
Now, assume that the lower bound of the time complexity in each iteration $X$ is $\Omega(n_X^2)$.
Clearly, this assumption does not improve the time complexity of $\EnumIS$. 
Thus, the following equation holds for the total time complexity of all the child iterations:
$\mybar{T}(X) = \Theta\left(\sum_{1 \le i \le n_X} (n_X - i + 1 - d_{G_i(X)}(v_i))^2\right)$. 

Next, we consider the lower bound of $n_X - i + 1 - d_{G_i(X)}(v_i)$ for each $1 \le i \le n_X$. 
Let $\sig V' = \inset{X \in \sig V}{n_X > 2q}$.
Since the size of the input of an iteration is smaller than the size of an ancestor, 
$\sig T' = \sig T[\sig V' \cup ch(\sig V')]$ forms a tree, where $ch(\sig V') = \bigcup_{X \in \sig V'}ch(X) \setminus \sig V'$. 
Thus, we can use the push out amortization technique to analyze this case. 
Since $G$ has no large clique,
the upper bound of $d_{G_i(X)}(v_i)$ is obtained from Lemma~\ref{lem:ub}.
This lemma can easily be derived from Theorem~\ref{theo:turan} that is shown by Tur{\'a}n. 
Let $\tau = (q - 1)/q$. 

\begin{theorem}[(Tur{\'a}n's Theorem~\cite{Turan:1941})]
\label{theo:turan}
    For any integer $q$ and $n$, 
    a graph $G$ that does not contain $K_q$ as a subgraph has at most $\frac{n^2\tau}{2}$ edges. 
\end{theorem}

\begin{lemma}
\label{lem:ub}
    Let $G$ be a graph and $v$ be a vertex with the minimum degree in $G$. 
    If the size of a maximum clique in $G$ is at most $q-1$, 
    then $d(v) \le n\tau$, where $n$ is the number of vertices in $G$. 
\end{lemma}

\begin{proof}
    If the minimum degree of $G$ is more than $n\tau$,  
    then $G$ has more than $\frac{n^2\tau}{2}$ edges. 
    This contradicts Theorem~\ref{theo:turan} and the statement holds.
\end{proof}

Using this upper bound, we show the following lemma which implies that if the size of the input graph of $X$ is large enough, 
that is, $n_X > 2q$, 
then the total computation time of all the child iterations of $X$ consumes more than that of $X$.

\begin{lemma}
\label{lem:lb}
    Let $X$ be an internal iteration in $\sig T'$. 
    There exists a constant $c>0$ such that 
    $\mybar{T}(X) > c(n_X(n_X + 1)(n_X + 2)/6q^2 - q/6 - 1)$. 
\end{lemma}

\begin{proof}
    Let $n_X$ be the number of vertices in $G(X)$ and $i$ be an integer. 
    If $i<n_X  - q$, 
    then $(n_X - i - d_{G_i(X)}(v_i))^2 > (n_X - i + 1 - (n_X - i + 1)\tau)^2$ from Lemma~\ref{lem:ub}. 
    Hence, $\mybar{T}(X) \ge c\sum_{1 \le i \le n_X - q - 1}((n_X - i + 1) - (n_X - i + 1)\tau)^2$
    since $((n_X - i + 1) - d_{G_i(X)}(v_i))$ is non negative for any $i$. 
    Therefore, 
    \begin{eqnarray*}
        \mybar{T}(X) 
        &  =  & c\sum_{1 \le i \le n_X}((n_X - i + 1) - d_{G_i(X)}(v_i))^2\\
        & \ge & c\sum_{1 \le i \le n_X - q - 1}((n_X - i + 1) - (n_X - i + 1)\tau)^2\\
        &  >  & cn_X(n_X + 1)(n_X + 2)/6q^2 - cq/6 - c
    \end{eqnarray*}
    holds. Thus, the statement holds.
\end{proof}

Using Lemma~\ref{lem:lb}, 
we can show that by choosing appropriate values for $\alpha$, $\beta$, and $T^*$,
any internal iteration of $\sig T'$ satisfies the PO condition. 

\begin{lemma}
\label{lem:po}
    Suppose that $\alpha = 3/2$, $\beta = 6$, and $T^* = cq$ for some positive constant $c$,  
    then, any internal iteration $X$ in $\sig T'$ with $n_X \ge 2q$ satisfies the PO condition, that is, 
    $\mybar{T}(X) \ge \alpha T(X) - \beta(n_X + 1)T^*$.  
\end{lemma}

\begin{proof}
    From Lemma~\ref{lem:lb}, 
    there exists $c$  such that 
    $\mybar{T}(X) > c(n_X(n_X + 1)(n_X + 2)/6q^2 - q/6 - 1)$ holds. 
    Hence, 
    \begin{equation}
    \label{eq:diff}
        \mybar{T}(X) - \alpha T(X) + \beta (n_X + 1)T^* \ge n_X(n_X + 1)(n_X + 2)/6q^2 - q/6 - 1 - 3n_X^2/2 + 6qn_X + 6q. 
    \end{equation}

    The right hand side of Eq.~(\ref{eq:diff}) is minimum when $n_X = 2q$ 
    since the side is monotonically increasing for $n_X \ge 2q$.
    Hence,
    \begin{align*}
    \label{eq:pushout}
        &      n_X(n_X + 1)(n_X + 2)/6q^2 - q/6 - 1 - 3n_X^2/2 + 6qn_X + 6q\\
        &  >   n_X(n_X + 1)(n_X + 2)/6q^2 - 3n_X^2/2 + 6qn_X\\
        & \ge  n_X\left((n_X + 1)(n_X + 2)/6q^2 - 3n_X/2 + 6q\right)\\
        &  >   2q(2/3 + 3q) > 0. 
    \end{align*}
    Therefore, any internal iteration $X \in \sig T'$ satisfies the PO condition and the statement holds.
\end{proof}

Recall that any leaf iteration $L$ receives $\order{\beta T(L)/(\alpha - 1)}$ computation time from the parent. 
Hence, the following lemma holds.

\begin{lemma}
\label{lem:leaf}
    Let $L$ be a leaf iteration of $\sig T'$ and $X$ be an internal iteration of $\sig T'$. 
    Then, $L$ receives $\order{T(L)}$ computational time and $X$ needs $\order{q}$ time. 
\end{lemma}

\begin{proof}
    From Lemma~\ref{lem:po}, any internal iteration in $\sig T'$ satisfies the PO condition. 
    Remind that 
    any leaf iteration $L$ receives at most $\order{\beta T(L)/(\alpha - 1)} = \order{T(L)}$ time since $\alpha$ and $\beta$ are positive constants. 
    In addition, each internal iteration $X$ of $\sig T'$ has $\beta T^*/(\alpha - 1) = \order{q}$. 
    Thus, the statement holds. 
\end{proof}

From Lemma~\ref{lem:po} and the distribution rule,
any internal iteration in $\sig T'$ has at most $\order{q}$ computation time on average. 
In addition, any leaf iteration $L$ in $\sig T'$ has at most $\order{T(L)}$ computation time. 
Note that from the definition, $T(L) = \order{n^2_L} = \order{ch(L)n_L} = \order{ch(L)q}$.
From Lemma~\ref{lem:smalltime} and Lemma~\ref{lem:leaf}, we can show the following main theorem.

\begin{theorem}
\label{theo:constant}
    \EnumIS enumerates all independent sets in $\order{q}$ amortized time even if the exact value of $q$ is unknown. 
\end{theorem}

\begin{proof}
From Lemma~\ref{lem:smalltime}, the time complexity of each iteration in $\sig T(X)$ is
$\order{n_X}$ time on average. 
From Lemma~\ref{lem:leaf}, 
$X$ receives at most $\order{T(X)}$ time from the parent. 
Hence, 
if $n_X \le 2q$,  
then from Lemma~\ref{lem:smalltime}, 
any descendant iteration of $X$ has $\order{q}$ computation time. 
If $n_X > 2q$, then from Lemma~\ref{lem:po} and the distribution rule of the computational cost, 
any $X$ has $\order{q}$ time on average. 
Note that 
the difference between $S(X)$ and $S(Y)$ is exactly one vertex for any iteration $X$ and its child iteration $Y$. 
Thus, the total size of what \EnumIS outputs is bounded by the number of iteration of \EnumIS. 
Therefore, 
by outputting only the difference between the $i$-th solution and the $i+1$-th solution instead of the $i+1$-th solution, 
the amortized time complexity of \EnumIS is $\order{q}$ time and the statement holds. 
\end{proof}

\section{A linear space implementation of \text{EIS}}

\newcommand{\zero}{\texttt{0}}
\newcommand{\one}{\texttt{1}}
\newcommand{\xzero}{x_{\zero}}
\newcommand{\xone}{x_{\one}}
\newcommand{\Mat}[1]{M_{#1}} 
\newcommand{\MR}[1]{M^r_{#1}} 
\newcommand{\Row}[2]{M_{#1}[#2]}
\newcommand{\RowR}[2]{M^r_{#1}[#2]}

\newcommand{\SR}[1]{Seq^r(#1)}

\newcommand{\SL}[1]{SL(#1)}
\newcommand{\SLmod}[1]{\widehat{SL}(#1)}
\newcommand{\SLdiff}[1]{Q(#1)}
\newcommand{\op}[1]{op(#1)}

\newcommand{\GiX}{X}
\newcommand{\GYi}{Y}

\newcommand{\LongInRowR}[3]{R^r(#1, #2, #3)}
\newcommand{\LongInRow}[3]{R(#1, #2, #3)}
\newcommand{\InRowR}[3]{R^r}
\newcommand{\InRow}[3]{R}

In this section, we show that we can implement \EnumIS in linear space. 
The main space bottleneck associated with \EnumIS relates to the following two points: 
One is the representation of input graphs. 
If we naively employ an adjacency matrix to represent an input graph, 
\EnumIS uses $\order{n^2}$ space. 
However, if we employ an adjacent list, then linear space is obtained 
but then it becomes difficult to obtain the input graph for a child iteration $Y$ of a current iteration $X$ in $\order{n_Y^2}$ from $G(X)$. 
Note that $G(Y)$ can easily be obtained in $\order{n_X^2}$ time. 
The other bottleneck is related to the smallest-last ordering. 
If each iteration of \EnumIS stores the smallest-last ordering, 
since the number of iterations between the root iteration and a leaf iteration is at most $n$, 
\EnumIS needs $\order{n^2}$ space.  
To overcome these difficulties and in particular, to achieve  $\order{n + m}$ space,  
we use \name{run-length encoding} for compressing an adjacency matrix and a \name{partial smallest-ordering} to
store only the differences between the smallest-last orderings.

We summarize the data structures that are stored during execution of an iteration $X$ as follows. 
Let $Z$ be an ancestor iteration of $X$ and $Z'$ be the parent of $Z$.  
Suppose that a vertex $z'_i$ is picked on $Z'$. 
We will provide  precise definitions of these data structures in the remainder of this section. 
Roughly speaking, $\MR{G}$ is the run-length encoded adjacent matrix, $\SL{Z}$ is the smallest-last ordering of $G(Z)$, 
and $\LongInRowR{G_i(Z')}{G(Z)}{z'_i}$ represents the removed vertices from $Z'$. 
\begin{enumerate}
    \item $\MR{G_i(X)}$ and the smallest-last ordering of $G_i(X)$, 
    \item Vertices from the first $i$ elements on $\SL{Z}$ and the position of $u$ on $\SL{Z'}$ for each  $u \in N[z'_i]$,
    \item $\SLdiff{\SLmod{Z', z'_i}, \SL{Z}}$ for restoring $\SL{Z'}$ for each $Z$, 
    \item $\RowR{G_i(Z)}{u}$ for each $u \in N_{G_i(Z)}[z_i]$ for a picked vertex $z_i$, and
    \item $\LongInRowR{G_i(Z')}{G(Z)}{z'_i}$ for each $u \in V \setminus N_{G_i(Z')}[z'_i]$.  
\end{enumerate}

\begin{figure}
    \centering
    \includegraphics[width=0.7\textwidth]{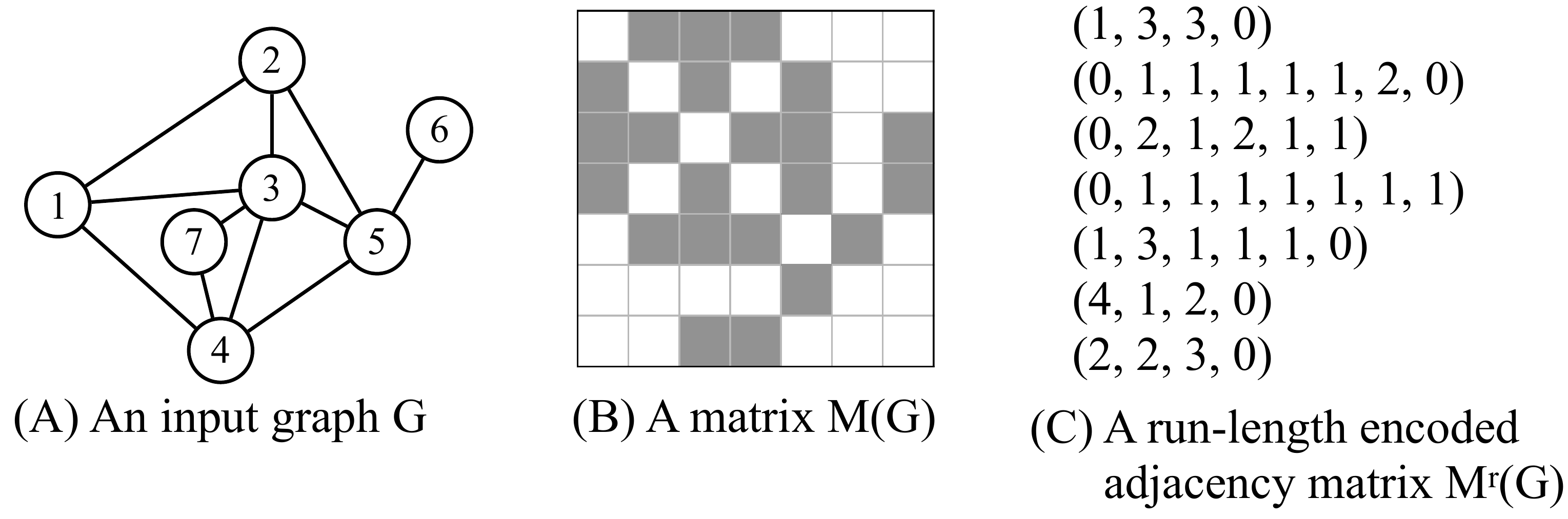}
    \caption{
    An example of an input graph $G$, its adjacency matrix, and its compressed representation. 
    A cell $(i, j)$ in $B$ is gray if $i$ and $j$ are adjacent. Otherwise, the cell is white.  
        (C) shows a run-length encoded adjacency matrix $\MR{G}$. }
    \label{fig:rl}
\end{figure}


First, we introduce the compression of input graphs by \name{run-length encoding}, which is a lossless data compression. 
We define  run-length encoding and run-length encoded adjacency matrix. 
Let $Seq$ be a sequence consisting of $\zero$ and $\one$. 
We define a \name{run-length encoded $\zero$-$\one$ sequence $Seq^r = (a_1, b_1, \dots, a_k, b_k)$} of $Seq$ as follows:  
Let $a_0 = b_0 = 0$. 
For $i > 0$,  $a_i$ is the length of the interval between consecutive $\zero$ sequences starting from the $\sum_{0 \le j \le i - 1} (a_j + b_j) + 1$-th element in $Seq$. 
Similarly, $b_i$ is the length of the interval between consecutive $\one$ sequences starting from the $\sum_{0 \le j \le i - 1} (a_j + b_j) + a_i + 1$-th element in $Seq$. 
For example, if $Seq = (\one, \one, \one, \zero, \zero, \zero, \one, \one, \zero)$, then $Seq^r = (0, 3, 3, 2, 1, 0)$. 
This is denoted by $\size{Seq^r} = 2k$; let us call this the \name{length} of $Seq^r$. 
The following lemma holds for the length of $Seq^r$. 

\begin{lemma}
\label{lem:01}
    Let $Seq$ be a $\zero$-$\one$ sequence and $Seq^r$ be a run-length encoded sequence of $Seq$. 
    Then, the length of $Seq^r = (a_1, b_1, \dots, a_k, b_k)$ is at most $\min\set{2\xzero, 2\xone} + 2$, 
    where $\xzero = \sum_{1 \le i \le k}a_i$ and $\xone = \sum_{1 \le i \le k}b_i$ is the number of $\zero$ and $\one$ in $Seq$, respectively. 
\end{lemma}

\begin{proof}
    Since $a_i > 0$ for $1 < i \le k$ and $b_i > 0$ for $1 \le i < k$, 
    $k$ is at most $\min\set{\xzero, \xone} + 1$. 
    Hence, the length of $Seq^r$ is at most $\min\set{2\xzero, 2\xone} + 2$ and the statement holds.
\end{proof}

Let $\Mat{G}$ be an adjacency matrix of $G$, 
$\Row{G}{j}$ be the $j$-th row of $\Mat{G}$, 
and let $\MR{G} = (\RowR{G}{1}, \dots, \RowR{G}{n})$ of $\Mat{G}$. 
We call $\MR{G}$ the \name{run-length encoded adjacency matrix} of $\Mat{G}$. 
An example of a run-length encoded adjacency matrix is shown in Figure~\ref{fig:rl}.
The next lemma shows that the size of $\MR{G}$ is linear in the size of $G$. 
Since $\Mat{G}$ is a $\zero$-$\one$ matrix, 
$\Row{G}{j}$ is a $\zero$-$\one$ sequence with length $n$ for $1 \le j \le n$. 

\begin{lemma}
\label{lem:rl}
    Let $\Mat{G}$ be an adjacency matrix of a graph $G$.
    Then, $\MR{G}$ needs $\order{n + m}$ space, where $n$ and $m$ are the number of vertices and the number of edges, respectively. 
\end{lemma}

\begin{proof}
    Since from Lemma~\ref{lem:01},  for each $1 \le j \le n$, 
    $\size{\RowR{G}{j}} \le \min\set{2\xzero, 2\xone} + 2$ and $\xone = d(v_j)$,  
    $\RowR{G}{j}$ consumes $\order{d(v_j)}$ space.
    Thus, $M^r(G)$ uses $\order{n + m}$ space and the statement holds. 
\end{proof}

\subsection{Generating the input graphs of child iterations}

In this subsection, 
we explain how to generate child iterations with the compressed inputs. 
This can be done in $\order{n_X^2}$ total time if \EnumIS uses $\MR{G(X)}$. 
To execute Line~\ref{algo:one:pick}, 
we first consider how to obtain the smallest-last ordering. 
If the graph is stored in adjacency list representation, the ordering can be obtained in $\order{n_X + m_X}$ time.  
Now, an adjacency list can be obtained from $\MR{G(X)}$ in $\order{n_X^2}$ time. 
Hence, we can compute the smallest-last ordering of $G(X)$ in $\order{n_X^2}$ time. 

Generating $\MR{G_{i+1}(X)}$ can be done in $\order{n_X}$ time since the run-length encoded sequences are ordered by the smallest-last ordering of $G(X)$, where $v_i$ is the first vertex in the sequences. 
Next, we consider how to generate $\MR{G(Y_i)}$ from $\MR{G_i(X)}$,  
where $Y_i$ is a child iteration of $X$ in Line~\ref{algo:1} whose input graph is obtained by removing $N[v_i]$ from $G_i(X)$. 
Our goal is $\order{n_X + n_{Y_i}^2}$ time for obtaining $\MR{G(Y_i)}$. 
If we can achieve this computation time, 
then the time complexity of $X$ can be bounded by $\order{n_X^2}$ by distributing $\order{n_{Y_i}^2}$ computation time from $X$ to its child $Y_i$. 
From the following lemma, 
we can compute $\MR{G(Y_i)}$ in $\order{n_X - i + \sum_{u \in V(Y_i)}\left(\size{\RowR{G_i(X)}{u}} + n_{Y_i}\right)}$ time. 

\begin{lemma}
\label{lem:right}
    Let $G = (V, E)$ be a graph and $v$ be the vertex with minimum degree in $G$. 
        Then, we can obtain $\MR{G\setminus N[v]}$ from $\MR{G}$ 
        in $\order{\size{V} + \displaystyle \sum_{w \in V \setminus N[v]}\left(\size{\RowR{G}{w}} + \size{\RowR{G}{v}}\right)}$ time. 
\end{lemma}

\begin{proof}
    By simply scanning $\RowR{G}{v}$ from the first element to the last element,
    we can obtain $V \setminus N[v]$ in $\order{\size{V \setminus N[v]} + \size{\RowR{G}{v}}} = \order{\size{V}}$ time. 
    Note that $V \setminus N[v]$ is sorted in the order of their indices. 
    
    Next, for each $w \notin N[v]$, 
    we compute $\RowR{G\setminus N[v]}{w}$ from $\RowR{G}{w}$. 
    Let $u_i$ be the $i$-th vertex in $V \setminus N[v]$,
    $c_j$ be the $j$-th element of $\RowR{G}{w}$,  
    and $C_j = \sum_{1 \le j' < j}c_{j'}$. 
    Let  $\RowR{G\setminus N[v]}{w} = \emptyset$ and $c^*$ be the last element of $\RowR{G\setminus N[v]}{w}$. 
    Now we can check whether $u_i$ and $w$ are connected or not by checking the following condition: 
    Let $j = 1$. 
    (1) If $j$ satisfies $C_j \le u_i < C_{j+1}$, then 
    (1.a) if the parity of $j$ and the index of $c^*$ is same then increment $c^*$ by one;  
    (1.b) if the parity of $j$ and the index of $c^*$ is not same then add $1$ as the last element to $\RowR{G\setminus N[v]}{w}$.  
    (2) If $j$ does not satisfy $C_j \le u_i < C_{j+1}$, 
    then increment the value of $j$ by one and check above two condition. 
    After updating the row, 
    we continue to check $u_{i+1}$ from $c_j$. 
    Now both $\RowR{G}{w}$ and $V \setminus N[v]$ are sorted.
    Hence, we can compute $\RowR{G\setminus N[v]}{w}$ from $\RowR{G}{w}$ by scanning $\RowR{G}{w}$ from the first element. 
    If we apply the above procedure naively, then we need $\order{\size{V}}$ time. 
    However, the above procedure can be done in $\order{\size{\RowR{G}{w}} + \size{\RowR{G}{w}}}$ time by 
    processing consecutive $\zero$s or $\one$s in one operation. 
    Hence, the statement holds.  
\end{proof}

Since $\order{\sum_{u \in V(G(Y_i))} n_{Y_i}} = \order{n_{Y_i}^2}$, 
we can push this computation time to $Y_i$. 
In addition, $X$ can also receive the computation time of $\order{n_X - i}$ without worsening the computation time of $X$. 
Next, 
we analyse $\order{\sum_{u \in V(G(Y_i))}\size{\RowR{G_i(X)}{u}}}$ more precisely.  

\begin{lemma}
\label{lem:rlm}
Let $G = (V, E)$ be a graph and $v$ be a vertex with the minimum degree in $G$. 
    Then, the length of $\RowR{G}{u}$ is $\order{\size{V\setminus N[v]}}$ for any $u \in V$. 
\end{lemma}
\begin{proof}
    Suppose that $u = v$. 
    Note that $\size{V \setminus N[v]}$ is equal to $\xzero$ in $\Row{G}{u}$. 
    Thus, from Lemma~\ref{lem:01}, the length of $\RowR{G}{u}$ is $\order{\size{V \setminus N[v]}}$. 
    Next, we assume that $u \neq v$. 
    Since $d(v) \le d(u)$, $\size{V \setminus N[u]} \le \size{V \setminus N[v]}$. 
    Hence, from Lemma~\ref{lem:01}, the length of $\RowR{G}{u}$ is at most $\min\set{2 \size{V \setminus N[u]}, 2d(u)} + 2 \le \min\set{2\size{V\setminus N[v]}, 2d(u)} + 2$.
    Therefore, $\size{\RowR{G}{u}} = \order{\size{V \setminus N[v]}}$. 
\end{proof}

From Lemma~\ref{lem:rlm}, $\size{\RowR{G_i(X)}{u}} = \order{n_{Y_i}}$ for any vertex $u$ in $V(G_i(X))$.   
Thus, $\order{\sum_{u \in V(G(Y_i))}\size{\RowR{G_i(X)}{u}}} = \order{n_{Y_i}^2}$ and 
we can push this computation time to $Y_i$. 
From the above discussion, we can compute $\MR{G(Y_i)}$ from $\MR{G_i(X)}$ without worsening the time complexity.

\subsection{Restoring the input graph of the parent iteration}

In this subsection, 
we consider backtracking from $Y_i$ to $X$.
The goal here is to show that the restoration of $\MR{G_i(X)}$ from $\MR{G(Y_i)}$ can be done 
in $\order{n_X + n_{Y_i}^2}$ time with $\order{n + m}$ total space. 
To restore $\MR{G_i(X)}$, 
we need to restore the smallest-last ordering in advance. 
Since we only consider the backtracking from $Y_i$ to $X$, 
if no confusion arises, 
we identify $Y_i$ with $Y$,  $G_i(X)$ with $\GiX$, and $G(Y)$ with $\GYi$.  

\subsubsection{Smallest-last orderings}
If \EnumIS stores the smallest-last ordering of $\GiX$ from that of $\GYi$ when making a recursive call and discards it when backtracking,  
then the total space is $\Omega\left(n^2\right)$ since the depth of the search tree  and the number of vertices of the input graphs can be $n$. 
Hence, to achieve $\order{n + m}$ space, \EnumIS does not entirely store the orderings, but rather stores them partially.
Let $\SL{X}$ and $\SL{Y}$ respectively be the smallest-last ordering of $\GiX$ and $\GYi$, 
and $v$ be the vertex such that $S(Y) = S(X) \cup \set{v}$. 
$\SLmod{X, v} = \SL{X}\setminus N_{\GiX}[v]$ denotes the \name{partial smallest-last ordering} obtained by removing the vertices in $N_{\GiX}[v]$ from $\SL{X}$. 
Let $SL$ be a smallest-last ordering and $SL[i]$ be the $i$-th vertex of $SL$. 
We say that a sequence $SL'$ is obtained by \name{shifting} $u$ at position $j$ to $p$ positions in $SL$ if 
$SL' = (SL[1], \dots, SL[j-p-1], SL[j], SL[j-p], \dots, SL[j-1], SL[j+1], \dots, SL[n])$, and write $SL' = \op{SL, u, p}$. 
This is refered to as a \name{shift operation}. 
Let  $\SLdiff{\SLmod{X, v}, \SL{Y}}$ be the sequence of pairs of a vertex and a shift value 
$((u_1, p_1), \dots, (u_\ell, p_\ell))$ with length $\ell$, 
such that   $\SL{Y} = \op{\op{ \cdots \op{\SLmod{X, v}, u_1, p_1} \cdots , u_{\ell-1}, p_{\ell-1}}, u_\ell, p_\ell}$. 
It can easily be shown that $\SLdiff{\SLmod{X, v}, \SL{Y}}$ can be obtained in $\order{\size{\SL{Y}}^2}$ time since
$p_j$ can be obtained in $\order{\size{\SL{Y}}}$ time for each $j \in [1, \ell]$.

\begin{lemma}
    \label{lem:sl}
    Let $X_*$ be the root iteration and $L$ be a leaf iteration. 
    Suppose that $\mathcal{I} = (X_* = I_1, \dots, I_i = L)$ is the path of iterations on $\sig T$
    such that for each $1 \le j < i$, $I_j$ is the parent of $I_{j+1}$ and $S(I_{j+1}) = S(I_j) \cup\set{v_j}$. 
    Then, 
    $\sum_{j \in [1, i-1]} \size{\SLdiff{\SLmod{I_j, v_j}, \SL{I_{j+1}}}} \le m$. 
\end{lemma}

\begin{proof}
    If a vertex $u$ is shifted, this implies that one of its incident edge is removed from a graph since each smallest-ordering is obtained by some fixed deterministic procedure. 
    Thus, the number of applying shift operations is at most $d_G(u)$ on $\mathcal{I}$. 
    Hence, the total number of applying shift operations is at most the number of edges in the input graph. 
\end{proof}

In addition, $\SL{X}$ is obtained by adding $N_{\GiX}[v]$ to $\SLmod{X, v}$ in $\order{\size{N_{\GiX}[v]}}$ time 
if for each $u \in N_{G(X)}[v]$, \EnumIS stores the position of $u$ in $\SL{X}$. 
This needs $\order{n}$ space in total since each vertex is removed from a graph at most once on the path from a current iteration to the root iteration. 
Hence, from the above discussion, we can obtain the following lemma: 

\begin{lemma}
    \label{lem:space}
    We can compute $\SL{X}$ from $\SL{Y}$ in $\order{n_X + n_Y^2}$ time with $\order{n+m}$ space in total when backtracking from $Y$ to $X$. 
\end{lemma}

From Lemma~\ref{lem:space}, \EnumIS demands $\order{n + m}$ space and $\order{n_X^2}$ time for each iteration $X$ on average for restoring the smallest-last orderings. 

\subsubsection{Run-length encoded adjacency matrices}

In this subsection, we demonstrate how to restore each row of $\MR{\GiX}$. 
Let $u$ be a vertex in $V(\GiX)$. 
Recall that $v$ is picked from $V(\GiX)$ and added to $S(Y)$. 
If  $u \notin V(\GYi)$, 
by just adding $\RowR{\GiX}{u}$ to $\RowR{\GYi}{u}$, 
we can restore $\RowR{\GiX}{u}$ since \EnumIS keeps $\RowR{\GiX}{u}$ until backtracking.
In addition, once $u$ is removed from $\GiX$, it will never appear in the input graph of a descendant iteration of $X$. 
Thus, \EnumIS requires linear space in total for storing removed $\RowR{\GiX}{u}$. 

Suppose $u \in V(\GYi)$.  
To restore $\RowR{\GiX}{u}$ by adding some vertices to $\GYi$, 
we use data structures $\LongInRow{\GiX}{\GYi}{u}$ and $\LongInRowR{\GiX}{\GYi}{u}$.
$\LongInRow{\GiX}{\GYi}{u}$ represents the neighborhood of $u$ that is removed from $\GiX$ to obtain $\GYi$, and
$\LongInRowR{\GiX}{\GYi}{u}$ is its run-length encoded representation. 
In the following,  
we fix $X$, $Y$, and $u$, and 
we abuse notation using $\LongInRow{\GiX}{\GYi}{u}$ and $\LongInRowR{\GiX}{\GYi}{u}$ to denote $\InRow{\GiX}{\GYi}{u}$ and $\InRowR{\GiX}{\GYi}{u}$, respectively.

Their precise definitions are as follows: 
$\InRow{\GiX}{\GYi}{u}$ is a $\zero$-$\one$ sequence with length $d_{\GiX}(v)$. 
The $j$-th element of $\InRow{\GiX}{\GYi}{u} = \one$ if the $j$-th neighbor of $v$ in $\GiX$ is adjacent to both $v$ and $u$ in $\GiX$. 
Otherwise, $\InRow{\GiX}{\GYi}{u} = \zero$. 
$\InRowR{\GiX}{\GYi}{u}$ is the run-length encoded sequence of $\InRow{\GiX}{\GYi}{u}$. 
By scanning $\Row{\GiX}{v}$, $\InRowR{\GiX}{\GYi}{u}$, and $\RowR{\GYi}{u}$, from the head simultaneously, 
we can efficiently obtain the removed vertices from $\RowR{\GiX}{u}$. 
Figure~\ref{fig:restore:gyu} provides an intuitive explanation of how this is achieved. 
The values of position 1, 2, 4, 6, and 7 of $\Row{\GiX}{u}$ come from $\InRow{\GiX}{\GYi}{u}$. %
The remaining values come from $\Row{\GYi}{u}$. 

\begin{figure}
    \centering
        \begin{tabular}{c|rrrrrrrrrrr}
             Position & 1      & 2      & 3      & 4      & 5      & 6      & 7      & 8      & 9      & 10 & 11\\
               \hline
            $\Row{\GiX}{v}$         & 1      & 1      & 0      & 1      & 0      & 1      & 1      & 0      & 0      & 0 & 0\\
            $\InRow{\GiX}{\GYi}{u}$ & $x_1$  & $x_2$  & $x_3$  & $x_4$  & $x_5$                                           \\
            $\Row{\GYi}{u}$           & $y_1$  & $y_2$  & $y_3$  & $y_4$  & $y_5$ & $y_6$                                          \\
            \hline
            $\Row{\GiX}{u}$         & $x_1$  & $x_2$  & $y_1$  & $x_3$  & $y_2$  & $x_4$  & $x_5$  & $y_3$  & $y_4$  & $y_5$ & $y_6$ \\
        \end{tabular}
        \caption{Restoring $\Row{\GiX}{u}$ using $\Row{\GiX}{v}$, $\InRow{\GiX}{\GYi}{u}$, and $\Row{\GYi}{u}$. }
    \label{fig:restore:gyu}
\end{figure}

\begin{lemma}
    \label{lem:restore:row}
    Let $v$ be a vertex with the minimum degree in $\GiX$, and $\GYi = \GiX \setminus N_{\GiX}[v]$. 
    We assume that the followings are given: 
    $\MR{\GYi}$, 
    for each $u \notin N_{\GiX}[v]$, 
    $\InRowR{\GiX}{\GYi}{u}$ and the position of $u$ on the smallest-last ordering of $\GiX$, and
    the smallest-last ordering of $\GYi$ with the shift operation sequence for the smallest-last ordering of $\GiX$. 
    Then, we can restore $\MR{\GiX}$ in $\order{n_X + n_Y^2}$ time 
    and $\order{\size{V(\GiX)}}$ space. 
\end{lemma}

\begin{proof}
    Let $u$ be a vertex in $V(\GYi)$. 
    From Lemma~\ref{lem:space}, we can obtain the smallest-last ordering of $\GiX$ in the claimed time complexity. 
    In addition, by scanning the smallest-last ordering of $\GiX$ from the head, we can reorder $\RowR{\GYi}{u}$.
    This reorder needs $\order{n_X + n_Y^2}$ time and $\order{\size{V(\GiX)}}$ space.
    
    Next, we consider how to restore $\RowR{\GiX}{u}$. 
    If we decode $\RowR{\GiX}{u}$ naively, then we need $\order{n_X}$ time for each $u$ and $\order{n_X^2}$ time in total. 
    This may exceed $\order{n_X + n_Y^2}$ time. 
    Thus, we employ efficient restoring procedure, in particular, without decoding. 

    We first consider the restoring procedure for $\zero$-$\one$ sequences, i.e., non-compressed sequences. 
    See Fig.~\ref{fig:restore:gyu} for a concrete example. 
    To compute $\Row{\GYi}{u}$, we first cut $\Row{\GiX}{u}$ into intervals and then concatenating some intervals. 
    $\InRow{\GiX}{\GYi}{u}$ was obtained by concatenating the remaining intervals that are not included in $\Row{\GYi}{u}$. 
    Hence, we can restore $\Row{\GiX}{u}$ by combining $\Row{\GYi}{u}$ and $\InRow{\GiX}{\GYi}{u}$. 
    This can be done by scanning $\Row{\GiX}{v}$ from the head. Note that $\Row{\GiX}{v}$ is stored when $v$ is added to a solution. 
    Moreover, both $\Row{\GYi}{u}$ and $\InRow{\GiX}{\GYi}{u}$ are sorted in the order of $\Row{\GiX}{v}$. 
    If a value of $\Row{\GiX}{v}$ is $\one$, then the corresponding value of $\Row{\GiX}{u}$ is recorded in $\InRow{\GiX}{\GYi}{u}$ since the corresponding vertex is adjacent to $v$ and now $v$ is in a solution. 
    Otherwise, the corresponding value is recorded in $\Row{\GYi}{u}$.
    Thus, if we meet $\one$ on $\Row{\GiX}{v}$, 
    then we put the corresponding value of $\InRow{\GiX}{\GYi}{u}$ to $\Row{\GiX}{u}$.
    Otherwise, then we put the corresponding value of $\Row{\GYi}{u}$ to $\Row{\GiX}{u}$.

    This idea can be also extended to the run-length encoded sequences. 
    Let $\Row{\GiX}{v} = (a_1, b_1, \dots, a_k, b_k)$. 
    First, we partition $\InRowR{\GiX}{\GYi}{u}$ into $A_1, A_2, \dots, A_{k'}$ by using $a_1, a_2, \dots, a_k$ as follows:  
    \begin{itemize}
        \item  $A_1 = (\InRowR{\GiX}{\GYi}{u}[1], \InRowR{\GiX}{\GYi}{u}[2], \dots, \InRowR{\GiX}{\GYi}{u}[j_1], \alpha_1)$  such that  \\
            $\displaystyle\sum_{j'=1,\dots, j_1} \InRowR{\GiX}{\GYi}{u}[j'] \le a_1 < \sum_{j'=1,\dots, j_1+1} \InRowR{\GiX}{\GYi}{u}[j']$ and 
            $\displaystyle a_1 = \alpha_1 + \sum_{j'=1,\dots, j_1} \InRowR{\GiX}{\GYi}{u}[j']$. 
        \item  $A_i = (\InRowR{\GiX}{\GYi}{u}[j_{i-1}+1]-\alpha_{i-1}, \InRowR{\GiX}{\GYi}{u}[j_{i-1}+2], \dots, \InRowR{\GiX}{\GYi}{u}[j_i], \alpha_i)$  such that  \\
            $\displaystyle \sum_{j'=j_{i-1}+1,\dots, j_i} \InRowR{\GiX}{\GYi}{u}[j'] \le a_i + \alpha_{i-1} < \sum_{j'=j_{i-1}+1,\dots, j_i+1} \InRowR{\GiX}{\GYi}{u}[j']$ and 
            $\displaystyle a_i = \alpha_i - \alpha_{i-1} + \sum_{j'=j_{i-1}+1,\dots, j_i} \InRowR{\GiX}{\GYi}{u}[j']$. 
    \end{itemize}
    By the similar way, 
    we also partition $\RowR{\GYi}{u}$ into $B_1, \dots, B_{k''}$ by using $b_1, b_2, \dots, b_k$. 
    Then, we alternatively put together $A_1, \dots A_{k'}$ and $B_1, \dots, B_{k''}$ into one sequence like the non-compressed version. 
    If the both of two consecutive values on the resultant sequence correspond to either $\one$ or $\zero$, 
    then we sum them up and replace it by the summation. 
    In addition, $k'$ and $k''$ can be bounded by $\order{n_Y}$. 
    Thus, for each $u \in \GYi$, this restoring can be done in $\order{n_Y}$ time. 
    In addition, for each $u \notin \GYi$, $\RowR{\GiX}{u}$ is stored when $X$ calls its child iteration. 
    Hence, we can restore $\MR{\GiX}$ in the claimed time complexity. 
    From Lemma~\ref{lem:rl} and $n_X \ge n_Y$, the above procedure only uses $\order{\size{V(\GiX)}}$ space. 
\end{proof}

Next, we consider the total space usage of $\InRowR{\cdot}{\cdot}{\cdot}$. 
From Lemma~\ref{lem:right}, we can easily see that we need $\order{d_{\GiX}(u) - d_{\GYi}(u)}$ space for storing $\InRowR{\GiX}{\GYi}{u}$ for $u$. 
However, if \EnumIS stores $\InRowR{\GiX}{\GYi}{u}$ such that $\InRow{\GiX}{\GYi}{u}$ does not contain $\one$, 
then the total space cannot be bounded by $\order{n + m}$ since each iteration has $\order{n_Y}$ additional space in the worst case and $\order{n^2}$ space in total. 
Hence,  we store $\InRowR{\GiX}{\GYi}{u}$ if $\InRow{\GiX}{\GYi}{u}$ contains at least one $\one$. 
These are stored by using a doubly linked list ordered in the smallest last ordering. 
Note that we now have the smallest-last ordering of $\GiX$ from the previous subsection. 
Thus, by scanning the list from the head, 
we can easily find vertices $u$ such that $\InRow{\GiX}{\GYi}{u}$ consists only of $\zero$s,  
and obtain $\InRow{\GiX}{\GYi}{u}$ by just filling  $\zero$s. 
Thus, \EnumIS requires linear space for computing the data structures. 
In addition, the restoration of $G_i(X)$ from $G_{i+1}(X)$ can be done in $\order{n_X}$ time.
Combining with the discussion in the previous section, we can obtain the following theorem.

\begin{theorem}
\label{theo:linear}
    \EnumIS enumerates all independent sets in $\order{q}$ time using $\order{n + m}$ space even if the exact value of $q$ is unknown, 
    where $n$ represents the number of vertices, $m$ represents the number of edges, and $q$ is the minimum number such that $G$ does not contain a clique with $q$ vertices. 
\end{theorem}

Finally, we obtain the following corollary from Theorem~\ref{theo:linear} 
since $K_q$ contains any graph with $q$ vertices as a subgraph, 
e.g., the graph class of $K_5$-free graphs contains planar graphs. 

\begin{corollary}
    Let $\sig R$ be a set of graphs and $\sig C$ be a graph class such that any graph in $C$ does not have a graph in $\sig R$ as a subgraph. 
    If $\sig R$ contains a graph with constant size, 
    then \EnumIS enumerates all independent sets in a given graph $G \in \sig C$ 
    with $\order{1}$ amortized time and $\order{n + m}$ space.
\end{corollary}


When \EnumIS makes an iteration $X$, the size of the input graph for $X$ is smaller than that of the parent iteration. 
This reducing procedure can be regarded as a \name{kernelization} technique for FPT algorithms~\cite{Niedermeier:IPL:2000}. 
Thus, combining this technique with push out amortization, 
we showed a good boundary on a recursion tree and demonstrated a sophisticated complexity analysis.  
Our future work will examine the possibility of applying our novel technique to other enumeration problems. 

\section*{Acknowledgements}
The authors thank A. Conte for helpful comments which led to improvements in this paper. This work was partially supported by JSPS KAKENHI Grant Number JP19J10761 and JP19K20350, and JST CREST Grant Number JPMJCR18K3 and JPMJCR1401, Japan.

\bibliographystyle{abbrv}
\bibliography{main.bbl}


\end{document}

%% file: macro.tex
\newcommand{\set}[1]{\{#1\}}
\newcommand{\inset}[2]{\{#1 \mid #2\}}
\newcommand{\name}[1]{\emph{#1}}

\newcommand{\size}[1]{\left|#1\right|}
\newcommand{\order}[1]{O\left(#1\right)} 

\newtheorem{lemma}{Lemma}
\newtheorem{theorem}[lemma]{Theorem}
 \newtheorem{corollary}[lemma]{Corollary}

\newcommand{\sig}[1]{\mathcal{#1}}

\makeatletter
\@ifpackageloaded{algorithm2e}{
\SetKwProg{Fn}{Function}{}{}
\SetKwProg{Procedure}{Procedure}{}{}
\SetKwProg{Subprocedure}{Subprocedure}{}{}
\SetKwComment{tcc}{//}{}
\SetKwFunction{Output}{Output}%
\SetKwFunction{EnumIS}{EIS}
\SetKwFunction{RecEIS}{RecEIS}
\SetKw{Continue}{continue} 
\SetKwInOut{AlgInput}{Input}
\SetKwInOut{AlgOutput}{Output}
\SetKwInOut{AlgPrecondition}{Pre-conditions}
\SetKwInOut{AlgInvariant}{Invariants}
}{}
\makeatother

%% file: main.bbl
\begin{thebibliography}{10}

\bibitem{Avis:Fukuda:DAM:1996}
D.~Avis and K.~Fukuda.
\newblock Reverse search for enumeration.
\newblock {\em Discrete Appl. Math.}, 65(1):21--46, 1996.

\bibitem{Beigel:SODA:1999}
R.~Beigel.
\newblock Finding maximum independent sets in sparse and general graphs.
\newblock In {\em Proc. {SODA} 1999}, pages 856--857. {ACM/SIAM}, 1999.

\bibitem{Bonamy:Defrain:Heinrich:Raymond:STACS:2019}
M.~Bonamy, O.~Defrain, M.~Heinrich, and J.-F. Raymond.
\newblock {Enumerating Minimal Dominating Sets in Triangle-Free Graphs}.
\newblock In {\em Proc. {STACS 2019}}, volume 126 of {\em LIPIcs}, pages
  16:1--16:12, Dagstuhl, Germany, 2019. Schloss Dagstuhl--Leibniz-Zentrum fuer
  Informatik.

\bibitem{Cohen:Kimefeld:Sagiv:JCSS:2008}
S.~Cohen, B.~Kimelfeld, and Y.~Sagiv.
\newblock Generating all maximal induced subgraphs for hereditary and
  connected-hereditary graph properties.
\newblock {\em J. Comput. Syst. Sci.}, 74(7):1147 -- 1159, 2008.

\bibitem{Alessio:Roberto:SPIRE:2017}
A.~Conte, R.~Grossi, A.~Marino, T.~Uno, and L.~Versari.
\newblock Listing maximal independent sets with minimal space and bounded
  delay.
\newblock In {\em Proc. {SPIRE 2017}}, volume 10508, pages 144--160. Springer,
  2017.

\bibitem{Alessio:Roberto:ICALP:2016}
A.~Conte, R.~Grossi, A.~Marino, and L.~Versari.
\newblock {Sublinear-Space Bounded-Delay Enumeration for Massive Network
  Analytics: Maximal Cliques}.
\newblock In {\em Proc. {ICALP} 2016}, volume~55 of {\em LIPIcs}, pages
  148:1--148:15. Schloss Dagstuhl--Leibniz-Zentrum fuer Informatik, 2016.

\bibitem{Conte:Grossi:Marino:Versari:SIAM:JoDM:2019}
A.~Conte, R.~Grossi, A.~Marino, and L.~Versari.
\newblock Listing maximal subgraphs satisfying strongly accessible properties.
\newblock {\em SIAM Journal on Discrete Mathematics}, 33(2):587--613, 2019.

\bibitem{Conte:Uno:STOC:2019}
A.~Conte and T.~Uno.
\newblock New polynomial delay bounds for maximal subgraph enumeration by
  proximity search.
\newblock In {\em Proceedings of the 51st Annual ACM SIGACT Symposium on Theory
  of Computing}, STOC 2019, pages 1179--1190, New York, NY, USA, 2019. ACM.

\bibitem{Eppstein:JGAA:2003}
D.~Eppstein.
\newblock Small maximal independent sets and faster exact graph coloring.
\newblock {\em J. Graph Algorithms Appl.}, 7(2):131--140, 2003.

\bibitem{Grohe:FSTTCS:2013}
M.~Grohe, S.~Kreutzer, and S.~Siebertz.
\newblock Characterisations of nowhere dense graphs (invited talk).
\newblock In {\em Proc. {FSTTCS} 2013}, volume~24 of {\em LIPIcs}, pages
  21--40. Schloss Dagstuhl--Leibniz-Zentrum fuer Informatik, 2013.

\bibitem{Johnson:Yannakakis:Papadimitriou:IPL:1988}
D.~S. Johnson, M.~Yannakakis, and C.~H. Papadimitriou.
\newblock On generating all maximal independent sets.
\newblock {\em Inf. Process. Lett.}, 27(3):119 -- 123, 1988.

\bibitem{Kashiwabara:JA:1992}
T.~Kashiwabara, S.~Masuda, K.~Nakajima, and T.~Fujisawa.
\newblock Generation of maximum independent sets of a bipartite graph and
  maximum cliques of a circular-arc graph.
\newblock {\em J. Algorithms}, 13(1):161--174, 1992.

\bibitem{Kazuhiro:ISAAC:2018}
K.~Kurita, K.~Wasa, H.~Arimura, and T.~Uno.
\newblock Efficient enumeration of dominating sets for sparse graphs.
\newblock In {\em Proc. {ISAAC} 2018}, volume 123 of LIPIcs, pages 8:1--8:13.
  Schloss Dagstuhl--Leibniz-Zentrum fuer Informatik, 2018.

\bibitem{Leung:JA:1984}
J.~Y. Leung.
\newblock Fast algorithms for generating all maximal independent sets of
  interval, circular-arc and chordal graphs.
\newblock {\em J. Algorithms}, 5(1):22--35, 1984.

\bibitem{Makino:Uno:SWAT:2004}
K.~Makino and T.~Uno.
\newblock {New Algorithms for Enumerating All Maximal Cliques}.
\newblock In {\em Proc. {SWAT} 2004}, volume 3111 of {\em LNCS}, pages
  260--272. Springer, 2004.

\bibitem{Manoussakis:TCS:2018}
G.~Manoussakis.
\newblock A new decomposition technique for maximal clique enumeration for
  sparse graphs.
\newblock {\em Theor. Comput. Sci.}, 2018.

\bibitem{Matula:Beck:JACM:1983}
D.~W. Matula and L.~L. Beck.
\newblock Smallest-last ordering and clustering and graph coloring algorithms.
\newblock {\em J. ACM}, 30(3):417--427, 1983.

\bibitem{Minty:JCT:1980}
G.~J. Minty.
\newblock On maximal independent sets of vertices in claw-free graphs.
\newblock {\em J. Comb. Theory, Ser. {B}}, 28(3):284--304, 1980.

\bibitem{Niedermeier:IPL:2000}
R.~Niedermeier and P.~Rossmanith.
\newblock A general method to speed up fixed-parameter-tractable algorithms.
\newblock {\em Inf. Process. Lett.}, 73(3-4):125--129, 2000.

\bibitem{Okamoto:JDA:2008}
Y.~Okamoto, T.~Uno, and R.~Uehara.
\newblock Counting the number of independent sets in chordal graphs.
\newblock {\em J. Discrete Algorithms}, 6(2):229--242, 2008.

\bibitem{Tomita:TCS:2006}
E.~Tomita, A.~Tanaka, and H.~Takahashi.
\newblock The worst-case time complexity for generating all maximal cliques and
  computational experiments.
\newblock {\em Theor. Comput. Sci.}, 363(1):28--42, 2006.

\bibitem{Tsukiyama:Ide:SICOMP:1977}
S.~Tsukiyama, M.~Ide, H.~Ariyoshi, and I.~Shirakawa.
\newblock A new algorithm for generating all the maximal independent sets.
\newblock {\em SIAM J. Comput.}, 6(3):505--517, 1977.

\bibitem{Turan:1941}
P.~Tur\'{a}n.
\newblock On an extremal problem in graph theory.
\newblock {\em Matematikai \'{e}s Fizikai Lapok (in Hungarian)}, 48, 1941.

\bibitem{Uno:WADS:2015}
T.~Uno.
\newblock {Constant Time Enumeration by Amortization}.
\newblock In {\em Proc. {WADS} 2015}, volume 9214 of {\em LNCS}, pages
  593--605. Springer, 2015.

\bibitem{Wasa:Arimura:Uno:ISAAC:2014}
K.~Wasa, H.~Arimura, and T.~Uno.
\newblock {Efficient Enumeration of Induced Subtrees in a K-Degenerate Graph}.
\newblock In {\em Proc. {ISAAC} 2014}, volume 8889 of {\em LNCS}, pages
  94--102. Springer, 2014.

\bibitem{Wasa:COCOON:2018}
K.~Wasa and T.~Uno.
\newblock Efficient enumeration of bipartite subgraphs in graphs.
\newblock In {\em Proc. {COCOON} 2018}, volume 10976 of LNCS, pages 454--466.
  Springer, 2018.

\end{thebibliography}
